\documentclass[10pt,twocolumn]{IEEEtran}
\usepackage{soul,xcolor}
\usepackage{bm}
\include{graphics}
\include{amsmath}
\usepackage[normalem]{ulem}
\usepackage{multirow,enumitem}
\usepackage{algorithm}
\usepackage{algpseudocode}
\usepackage{verbatim}
\usepackage[latin1]{inputenc}
\usepackage{amsmath}
\usepackage{amsfonts}
\usepackage{amssymb}
\usepackage{mathrsfs}
\usepackage{booktabs}
\usepackage{url}
\usepackage{graphicx}
\usepackage[skip=1pt,font=scriptsize]{subcaption}
\usepackage{ifthen}
\usepackage{xspace}
\usepackage{dsfont}
\usepackage{bbm}
\usepackage{cite}
\usepackage{epsfig}
\usepackage{array}
\usepackage{multicol}
\usepackage{algpseudocode}
\usepackage{algorithm}
\usepackage{amssymb}
\usepackage{breqn}
\usepackage{esint}
\usepackage[skip=2pt,font=footnotesize]{caption}
\usepackage{multicol}
\usepackage{amsthm}
\usepackage{setspace}
\usepackage{lipsum}
\usepackage{dblfloatfix}
\theoremstyle{plain}
\newtheorem{theorem}{Theorem}
\newtheorem{lemma}{Lemma}
\theoremstyle{definition}

\theoremstyle{remark}

\usepackage{footnote}
\newcommand{\numberthis}{\addtocounter{equation}{1}\tag{\theequation}}

\newcommand{\vect}[1]{\mathbf{#1}}

\newcommand{\wtot}{W_{\mathrm{tot}}}
\newcommand{\prob}{\mathbb{P}}
\DeclareMathOperator{\Exp}{\mathbb{E}}
\newcommand{\tfrm}{T_{\mathrm{fr}}}

\title{Second-best Beam-Alignment via\\ Bayesian Multi-Armed Bandits}

\author{Muddassar Hussain, Nicol\`{o} Michelusi
\thanks{This research has been funded by NSF under grant CNS-1642982.}
\thanks{Authors are with the School of Electrical and Computer Engineering, Purdue University. email: \{hussai13,michelus\}@purdue.edu.}
\vspace{-7mm}
}
\begin{document}

%\author{\IEEEauthorblockN{\IEEEauthorrefmark{1},
%\IEEEauthorrefmark{1}, \IEEEauthorrefmark{1}
%\IEEEauthorblockA{\IEEEauthorrefmark{1} \\
%}
%\IEEEauthorblockA{\IEEEauthorrefmark{2}\\
%}
%\IEEEauthorblockA{\IEEEauthorrefmark{3}
%\\
%}
%Email: @purdue.edu, @purdue.edu, @purdue.edu}
%}

\maketitle
\thispagestyle{empty}
\pagestyle{empty}
\setulcolor{red}
\setul{red}{2pt}
\setstcolor{red}
\begin{abstract}
Millimeter-wave (mm-wave) systems rely on narrow-beams to cope with the severe signal attenuation in the mm-wave frequency band. However, susceptibility to beam mis-alignment due to mobility or blockage requires the use of beam-alignment schemes, with huge cost in terms of overhead and use of system resources. In this paper, a beam-alignment scheme is proposed based on Bayesian multi-armed bandits, with the goal to maximize the alignment probability and the data-communication throughput. A Bayesian approach is proposed, by considering the state as a posterior distribution over angles of arrival (AoA) and of departure (AoD), given the history of feedback signaling and of beam pairs scanned by the base-station (BS) and the user-end (UE). A simplified sufficient statistic for optimal control is identified, in the form of preference of BS-UE beam pairs. By bounding a value function, the \emph{second-best preference} policy is formulated, which strikes an optimal balance between exploration and exploitation by selecting the beam pair with the current \emph{second-best} preference. Through Monte-Carlo simulation with analog beamforming, the superior performance of the \emph{second-best preference} policy is demonstrated in comparison to existing schemes based on \emph{first-best preference}, linear Thompson sampling, and upper confidence bounds, with up to 7\%, 10\% and 30\% improvements in alignment probability, respectively.
\end{abstract}
\begin{IEEEkeywords}
Millimeter-wave, beam-alignment, multi-armed bandits, Markov decision process
\end{IEEEkeywords}
\vspace{-2mm}
\section{Introduction}
Millimeter-wave (mm-wave) technology has emerged as a promising solution to meet the demands of future communication systems supporting high capacity and mobility, thanks to abundant bandwidth availability \cite{channel_model}. 
However, high isotropic path loss and sensitivity to blockages pose challenges in the design of these systems \cite{rappaport_mmwave_book}.
 To overcome the severe signal attention, 
mm-wave systems leverage narrow-beam communications, by using large antenna arrays at base stations (BSs) and user-ends (UEs).
However, narrow beams are highly susceptible to mis-alignment due to mobility and blockage, hence they require utilization of beam-alignment schemes,
which may cause huge overhead. 
\par 
Therefore, the design of beam-alignment schemes
with minimal overhead is of paramount importance, and has been a subject of intense research.   One of the earliest  yet most popular schemes is \emph{exhaustive search} \cite{exhaustive}, which scans sequentially through all possible BS-UE beam pairs and selects the one with maximum signal power for data communications.
To reduce the delay of exhaustive search, \emph{iterative} search is proposed in \cite{iterative}, where scanning is first performed using wider beams, followed by refinement using narrow beams. In the aforementioned heuristic schemes, the optimal design is not considered. To address this challenge, in our previous papers \cite{ita2017,ita2018,icc2018,TWC2019}, we considered the optimal design of interactive beam-alignment protocols that utilize  1-bit feedback from UEs.
 In \cite{ita2017,ita2018}, we design a throughput-optimal beam-alignment scheme for a single UE and two UEs, respectively, and we prove the optimality of a \emph{bisection search};
in \cite{icc2018}, we optimize the trade-off between data communication and beam-sweeping  in a mobile scenario where the BS widens its beam to mitigate the uncertainty on the UE position;  in \cite{TWC2019}, we incorporate the energy cost of beam-alignment, and prove the optimality of a \emph{fractional search} method. \emph{In our aforementioned papers \cite{ita2017,ita2018,icc2018,TWC2019},  the optimal design is carried out under the restrictive assumption of error-free single-bit feedback. However, this assumption may not hold in the presence of significant side-lobe gain and/or low signal-to-noise ratio (SNR)}.

The case of erroneous or noisy feedback is considered in recent work \cite{MAB,Javidi}, and our work \cite{allerton2018}. 
A coded beam-alignment scheme is proposed  in \cite{allerton2018} to correct these errors, but with no consideration of feedback to improve beam-selection. A multi-armed bandit (MAB) formulation based on upper confidence bound (UCB) is proposed in \cite{MAB}, by selecting the beam based on the empirical SNR distribution.
A hierarchical beam-alignment scheme based on posterior matching is proposed in \cite{Javidi}: therein, a \emph{first-best} policy is formulated, which selects the most likely beam pair based on the posterior distribution on the AoA-AoD pair. However, as we will see numerically, both UCB and first-best policies are prone to errors due to under-exploration of the beam space. 

 \par In this paper, we propose a beam-alignment design with the goal to maximize the alignment probability and the average throughput during the data communication phase. We pose the problem as a Markov decision process (MDP), where the beam pair is chosen based upon the \emph{belief} over the AoA-AoD pair, given the history of scanned beams and the received signal power. We identify a simplified sufficient statistic in the form of preference of the AoA-AoD beam pairs. We derive lower and upper bounds to the value function,
  based on which we propose a heuristic policy which selects the beam pair with the \emph{second-best} preference. We show numerically that this policy strikes a favorable trade-off between exploration and exploitation: instead of greedily choosing the beam corresponding to the most likely AoA-AoD pair (\emph{first-best} \cite{Javidi}), it chooses the second most likely one, leading to better exploration; at the same time, it avoids wasting precious resources to scan unlikely beam pairs, leading to better exploitation than other MAB techniques, such as linear Thompson sampling (LTS) \cite{Sutton} and UCB \cite{MAB}.
   The proposed \emph{second-best} scheme is shown to outperform
   \emph{first-best} \cite{Javidi}, LTS-based \cite{Sutton} and UCB-based \cite{MAB} schemes by up to 7\%, 10\% and 30\% in alignment probability, respectively. 
 \par The rest of the paper is organized as follows. In Sec.~\ref{sec:sys_mod}, we present the system model. In Sec.~\ref{sec:prob}, we formulate the problem and our proposed solution strategy. In Sec.~\ref{sec:numer}, we present numerical results, followed by final remarks in Sec.~\ref{sec:conc}. 
% \vspace{-2mm}
\section{System Model}
\label{sec:sys_mod}
We consider a downlink scenario with one BS and one UE, as depicted in Fig.~\ref{fig:sys_mod}. Time is divided into frames of duration
$\tfrm{=}T_{\mathrm{s}}N$, each with $N$ slots of duration
 $T_{\mathrm{s}}$. The frame is partitioned into two phases: a \emph{beam-alignment phase} of duration $LT_{\mathrm{s}}$ ($L{<}N$ slots), followed by a \emph{downlink data communication phase}, of duration $(N{-}L)T_{\mathrm{s}}$. Each beam-alignment slot  is further partitioned into a \emph{pilot transmission phase}, of duration $T_{\mathrm{pt}}$, followed by a \emph{feedback phase}, of duration $T_{\mathrm{fb}}$, with $T_{\mathrm{s}}{=}T_{\mathrm{pt}}{+}T_{\mathrm{fb}}$. These are detailed next.

\begin{figure}[!t]
\centering
\includegraphics[width=0.55\linewidth]{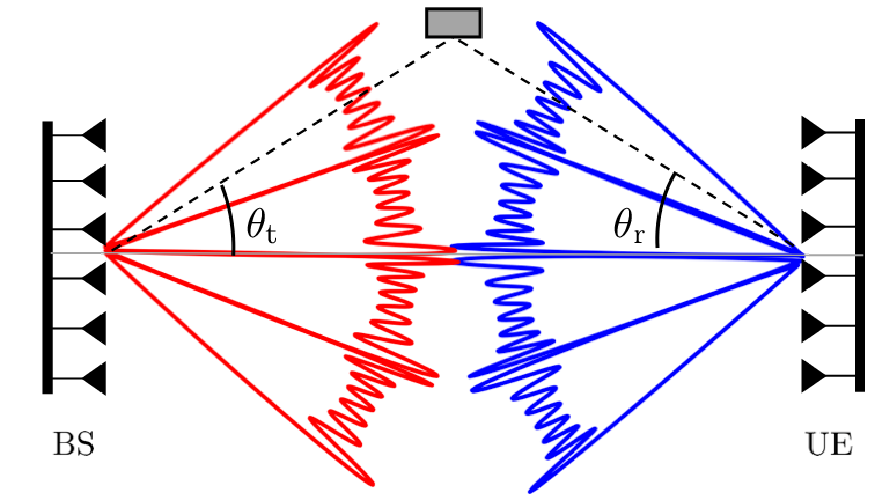}
\caption{System model; $M_{\mathrm t} = M_{\mathrm r} =  128$; beamforming algorithm in \cite{noh}. 
}
\label{fig:sys_mod}
\vspace{-6mm}
\end{figure}

The BS and UE are equipped with uniform linear arrays (ULAs) with $M_{\mathrm t}$ and $M_{\mathrm r}$ antenna elements, respectively, and use analog beamforming. The signal received at the UE  is 
\begin{align}
 \label{eq:receivedsignal}
 \vect z_{k} = \sqrt{P_{\mathrm{tx},k}}\vect u_{k}^H \vect H_{k} \vect v_{k}  \vect s + {\vect w_{k}},\ \forall k \in\{0,1,\dots,N-1\},
\end{align}
 where $P_{\mathrm{tx,k}}$ is the average transmit power of the BS;
 {$\mathbf s\in\mathbb C^{S}$}
  is the transmitted signal with $S$ symbols with $\Exp[\lVert\vect s\rVert_2^2]={S}$;
$\vect H_{k}{\in}\mathbb C^{M_{\mathrm r} \times M_{\mathrm t}}$ is the channel matrix; $\vect v_{k}{\in}\mathbb C^{M_{\mathrm t}}$ is the BS beamforming vector with
$\|\vect v_{k}\|_2^2=1$; $\vect u_{k}{\in}\mathbb C^{M_{\mathrm r}}$ is the UE combining vector with $\| \vect u_{k}\|_2^2=1$; $\vect w_{k}{\sim}\mathcal{CN}(\boldsymbol 0,N_0W_{\mathrm{tot}} \vect I)$ is additive white Gaussian noise (AWGN), with one-sided power spectral density  $N_0$ and system bandwidth $W_{\mathrm{tot}}$.
\\
%%%%%%%%%%%%
\textbf{Channel Model:} We use the extended Saleh-Valenzuela geometric model with a single-cluster \cite{extended}, as adopted in several previous works (e.g., see \cite{zorzi,jeffery_andrew_ba,TWC2019}). In fact, typical mm-wave channels have been shown to exhibit one dominant cluster containing most of the signal energy \cite{rappaport_channel_model}.  The single-cluster channel is modeled as
\begin{align}
\label{channel}
\vect H_k = \alpha_k \vect a_{\mathrm r}(\theta_{\mathrm r,k})\vect a_{\mathrm t}^H(\theta_{\mathrm t,k}), 
\end{align}
where  $\theta_{k}{\triangleq}(\theta_{\mathrm r,k},\theta_{\mathrm t,k}){\in}[-\frac{\pi}{2},\frac{\pi}{2}]^2$ is the angle of arrival (AoA) and angle of departure (AoD) pair associated to the dominant cluster, with complex fading gain $\alpha_{k}$; $\vect a_{\mathrm r}$ and $\vect a_{\mathrm t}$ are the UE and BS array response vectors, respectively, defined as
\begin{align*}
\vect a_{\mathrm x }(\theta_{\mathrm x }){=}\frac{1}{\sqrt{M_{{\mathrm x}}}}\left[1, e^{j\frac{2\pi d_{\mathrm x}}{\lambda}\psi_{\mathrm x}},\cdots, e^{j(M_{{\mathrm x}}-1)\frac{2\pi d_{\mathrm x}}{\lambda}\psi_{\mathrm x}}\right]^{\top}\!\!\!\!\!,
\mathrm x{\in}\{{\mathrm t },{\mathrm r}\},
\end{align*}
where $\psi_{\mathrm x}{=}\sin\theta_{\mathrm x }$, $d_{\mathrm x}$ is the antenna spacing, $\lambda{=}c/f_c$ is the wavelength at carrier frequency $f_c$, $c$ denotes the speed of light.
  We assume that during the duration of one frame $\tfrm$, $\theta_k$ remains unchanged,  $\theta_k{=}\theta$,  and $\alpha_k$ are i.i.d. Reyleigh fading in each slot with distribution $\alpha_k{\sim}\mathcal{CN}(0,\ell(d)^{-1})$, 
  where $\ell(d)$ is the path loss at distance $d$ from the BS. In fact, 
  %the channel gain coherence time is much shorter than the beam coherence time, implying that
   the AoA-AoD pair change much slower than the channel gain \cite{beam_coherence_time}. \\
  \textbf{Codebook structure:}
 In slot $k$, the BS uses the beamforming vector $\vect v_k{\in}\mathcal V$  and the UE uses the combining vector $\vect u_k{\in}\mathcal U$, from the codebooks $\mathcal V$ and $\mathcal U$, respectively. 
 We assume a sectored model \cite{TWC2019}, in which the AoA and AoD spaces are partitioned into sectors of equal beamwidth
 (as shown in Fig.~\ref{fig:sys_mod} for the case of four sectors, this model approximates well analog beamforming).
 Accordingly, 
 let $\mathcal B_{\mathrm r}(\vect u){\subseteq}[-\frac{\pi}{2},\frac{\pi}{2}]$ and $\mathcal B_{\mathrm t}(\vect v){\subseteq}[-\frac{\pi}{2},\frac{\pi}{2}]$
   denote the AoA and AoD supports of the UE combiner and BS beamformer vectors $\vect u{\in}\mathcal U$ and $\vect v{\in}\mathcal V$, respectively,
   with equal beamwidth $|\mathcal B_{\mathrm r}(\vect u)|{=}\frac{\pi}{|\mathcal U|},{\forall}\vect u{\in}\mathcal U$ and  $|\mathcal B_{\mathrm t}(\vect v)|{=}\frac{\pi}{|\mathcal V|},{\forall}\vect v{\in}\mathcal V$,  where $|\mathcal B|$ denotes the measure $|\mathcal B|{\triangleq}\int_{\mathcal B}\mathrm d x$.
    We define $\mathcal B(\vect u, \vect v){\triangleq}\mathcal B_{\mathrm r}(\vect u){\times}\mathcal B_{\mathrm t}(\vect v)$ as the joint AoA-AoD support of $(\vect u,\vect v)$. We assume that the angular supports are mutually orthogonal and form a partition of the entire AoA-AoD space $[-\frac{\pi}{2},\frac{\pi}{2}]^2$, i.e., $\mathcal B(\vect u,\vect v){\cap}\mathcal B(\tilde{\vect u},\tilde{\vect v}){=}\emptyset,\forall (\vect u,\vect v){\neq}(\tilde{\vect u},\tilde{\vect v})$ and
 $\cup_{\vect u\in\mathcal U}\mathcal B_{\mathrm r}(\vect u){=}
 \cup_{\vect v\in\mathcal V}\mathcal B_{\mathrm t}(\vect v){=}
 [-\frac{\pi}{2},\frac{\pi}{2}]$.
 Let $(\vect u^{(i)},\vect v^{(i)}),i\in \mathcal I{\triangleq}\{1,2,\ldots,|\mathcal U || \mathcal V| \} $ be any ordering of
 combining and beamforming vectors, and 
 $\mathcal B^{(i)} \triangleq \mathcal B(\vect u^{(i)},{\vect v}^{(i)})$ be their support.
  Let $A_k{\in}\mathcal I$ be the beam index of the  combining and beamforming vectors scanned in slot $k$, so that $(\vect u_k,\vect v_k){=}(\vect u^{(A_k)},\vect v^{(A_k)})$.
 Let $X$ be a discrete random variable denoting the index of the support that the AoA-AoD pair $\theta$ of the channel belongs to, so that $\theta{\in}\mathcal B^{({X})}$.
 Then, {from \eqref{eq:receivedsignal}-\eqref{channel}, the received signal} can {be} expressed as\footnote{The phase of
  $\vect u_{k}^H\vect a_{\mathrm r}(\theta_{\mathrm r})\vect a_{\mathrm t}^H(\theta_{\mathrm t})\vect v_{k}$
   is incorporated {into} $\alpha_k$.}
 \begin{align}
\vect z_k \approx \sqrt{P_{\mathrm{tx,k}}}\alpha_k \Big[(\sqrt{G}-\sqrt{g}){\delta[A_k,X]}{+} \sqrt{g}\Big]\vect s + \vect w_k,
 \end{align}
where
$\delta[\cdot]$ is the Kronecker's delta function, equal to $1$ if alignment is achieved ($A_k{=}X$), equal to $0$ otherwise ($A_k{\neq} X$);
 $G$ and $g$ are{, respectively,} the main and side lobe gains {of the sectored model}, expressed as
  $$
   G= \min_{(\theta_{\mathrm r},\theta_{\mathrm t})\in \mathcal B(\vect u^{(i)},\vect v^{(i)} )}|\vect a_{\mathrm r}(\theta_{\mathrm r})^H\vect u^{(i)}|^2
 |\vect a_{\mathrm t}^H(\theta_{\mathrm t})\vect v^{(i)}|^2{,\ \forall i,}
 $$
  $$
   g= \max_{(\theta_{\mathrm r},\theta_{\mathrm t})\not \in \mathcal B(\vect u^{(i)},\vect v^{(i)} )} |\vect a_{\mathrm r}(\theta_{\mathrm r})^H\vect u^{(i)}|^2
 |\vect a_{\mathrm t}^H(\theta_{\mathrm t})\vect v^{(i)}|^2{,\ \forall i.}
 $$
In the following, we describe the beam-alignment  and data communication procedures.
\\
 \textbf{Beam-Alignment:} In each slot $k$ of the beam-alignment phase, the BS transmits a pilot sequence $\vect s$ using the beam index $A_k$, with transmit power $P_{\mathrm{tx},k}{=}P_{\mathrm{ba}}$.
  Upon receiving $\vect z_k$ (based on the combining vector with index $A_k$), the UE uses a matched filter to compute the signal strength and sends the 
  \emph{normalized received power}
  feedback signal $Y_k$ back to the BS, of the form
 \begin{align}
   Y_k = 
   \frac{|\vect s^H\vect z_k|^2}{\|\vect s\|^2 N_0 \wtot(1+\Lambda g)},
   \end{align} 
   where $\Lambda{\triangleq}\frac{P_{\mathrm{ba}} \|\vect s\|^2}{N_0 \wtot \ell(d)}$ is the pre-beamforming receive SNR during beam-alignment.
Then, the probability density function (pdf) of $Y_k$ conditional on $(X,A_k){=}(x,a_{\mathrm s})$
 is given by
  \begin{align*}
  \label{eq:feedback_pdf}
  &f(Y_k{=}y|X{=}x;A_k{=}a_{\mathrm s}){=}\left[\nu e^{-\nu y} \right]^{\delta[a_{\mathrm s},x]}[e^{-y}]^{1-\delta[a_{\mathrm s},x] },\numberthis
  \end{align*}
  where $1/\nu$ is the mean signal power in case of alignment, with
  \begin{align}
  \nu \triangleq \frac{1+g\Lambda}{1+G\Lambda}.
  \end{align}
  The BS uses a Bayesian approach to select $A_k$: 
starting from $\mathcal H_0{\triangleq}\emptyset$ and
  given
 the history of feedback and scanned beam indices {$\mathcal H_{k}{\triangleq}\{ (A_j,Y_j)\}_{j=0}^{k-1}$}, the next beam index {$A_k$} is selected. This procedure continues until the end of the beam-alignment phase.\\
\textbf{Data communication:} Upon completion of the beam-alignment phase, given the history of feedback and actions $\mathcal H_{L}$, the BS selects the data communication parameters: beam index for data communication $A_{\mathrm d}{{\in}\mathcal I}$, transmission power $P_{\mathrm d}{\in}[0,P_{\max}]$, and data rate  $R_{\mathrm d}{\geq}0$. These  parameters are used until the end of the data communication phase.
 
 Let $b_0[x]$
 be the \emph{prior belief} over $X{=}x$ (or equivalently over $\theta{\in}\mathcal B^{(x)}$) available at the beginning of the beam-alignment phase.
We define the expected rate during the communication phase (normalized by the frame duration), as
\begin{align*}
\label{xcfcx}
&\bar R(A_{\mathrm d},P_{\mathrm d},R_{\mathrm d}|b_0,\mathcal H_{L})
\\&\qquad
\triangleq \frac{\tfrm-L T_{\mathrm{s}}}{\tfrm} \ 
\mathbb P \left(X=A_{\mathrm d}\Bigl | b_0,\mathcal H_{L}\right)
\hat R(R_{\mathrm d},P_{\mathrm d}), \numberthis
\end{align*}
where we have defined
\begin{align}
\label{dfgh}
&\hat R(R_{\mathrm d},P_{\mathrm d})
{\triangleq}R_{\mathrm d} \prob\left[ R_{\mathrm d}{\leq}\wtot \log_2 \left( 1+ \frac{|\alpha_k|^2 P_{\mathrm d} G}{N_0\wtot}\right)  \right]. \numberthis
\end{align} 
The  probability term in \eqref{xcfcx} is the probability of achieving correct alignment, given the prior $b_0$ and the history $\mathcal H_{L}$ of feedback and actions  during the beam-alignment phase,
whereas the
 probability term in \eqref{dfgh}
 denotes the probability of non-outage with respect to the realization of the fading process (i.i.d. over time), given that correct alignment has been achieved (we assume that mis-alignment yields outage with probability one, since $g\ll G$).
 \vspace{-3mm}
\section{Problem Formulation and Solution}
\label{sec:prob}
We now formulate the beam-alignment and data communication problem in the context of a decision process.
We define a policy $\mu$, part of our design, which operates as follows.
At time $k$ during beam-alignment,
given the history of feedback and actions $\mathcal H_k$, the BS
selects the \emph{beam-alignment action} $A_k{=}a_{\mathrm s}{\in}\mathcal I$ with probability  $\mu_k(a_{\mathrm s}|\mathcal H_k)$; given $\mathcal H_{L}$, the BS selects the 
data communication parameters as $(A_{\mathrm d},P_{\mathrm d},R_{\mathrm d})=\mu_{\mathrm d}(\mathcal H_{L})$.
The goal is to design $\mu$ so as to maximize 
the expected communication rate, i.e.,
\begin{align*}
\textbf{P0: }\max_\mu \mathbb E_\mu\left[\bar R(A_{\mathrm d},P_{\mathrm d},R_{\mathrm d}|b_0,\mathcal H_{L})|b_0\right],
\end{align*}
where the expectation $\mathbb E_\mu$ is conditional on the prior belief $b_0$ and on the
policy $\mu$ being executed during beam-alignment and data communication. 
Note that, using \eqref{xcfcx}, we can rewrite the optimization problem as
\begin{align*}
\textbf{P1: }&\max_\mu \mathbb E_\mu\left[
\mathbb P \left(X=A_{\mathrm d}\Bigl | b_0,\mathcal H_{L}\right)\Biggr|b_0\right]
\\&\times
\frac{\tfrm-L T_{\mathrm{s}}}{\tfrm}
\max_{R_d\geq 0,0 \leq P_d \leq P_{\max}}
\hat R(R_{\mathrm d},P_{\mathrm d})
,
\end{align*}
i.e., the problem can be decomposed into the following two independent problems:
1) find the optimal rate and power $(R_d^*,P_d^*)$ that maximize the expected rate in the communication phase, conditional on correct alignment being achieved ($X{=}A_{\mathrm d}$);
 2) find the optimal beam-alignment policy and the beam index for communication 
 $A_{\mathrm d}$ so as to maximize the probability of correct alignment.
 The first problem can be solved efficiently by maximizing \eqref{dfgh}.
 In the sequel, we consider the latter problem.

Let $b_k[x]{\triangleq} \prob (X= x|\mathcal H_k,b_0)$ be the belief over $X{=}x$ given the history of actions and feedback and prior belief $b_0$. It serves as  a sufficient statistic for optimal control for problem \textbf{P1}. In the following lemma, we present an equivalent simplified sufficient statistic along with its dynamics. 
\begin{lemma} Let $m_0[x]{{\triangleq}}\ln b_0[x]$ denote the prior preference of $X=x$. Given the action and feedback pair $(A_k,Y_k)$, 
 the belief at $k+1$ is updated as
\begin{align}
\label{eq:beliefup}
b_{k+1}[x] = \frac{\exp\{m_{k+1}[x]\}}{\sum_{l \in \mathcal I}\exp\{m_{k+1}[l]\} },
\end{align}
where
\begin{align}
\label{eq:mkup}
m_{k+1}[x] = m_k[x] + J(Y_k)\delta[A_k,x],\ \forall x\in \mathcal I,
\end{align}
and we have defined
\begin{align}
\label{eq:Jfun}
J(y) \triangleq (1-\nu)y +\ln \nu.
\end{align}

\end{lemma}
\begin{proof}
Given the belief $ b_k$ and $(A_k,Y_k)=(a_{\mathrm s},y)$, we have
\begin{align*}
\label{eq:beliefder}
&b_{k+1}[x]\stackrel{(a)}{=} \mathbb P(X=x|\mathcal H_{k+1})\\
&\stackrel{(b)}{\propto} f(Y_k= y|X = x, A_k =a_{\mathrm s} ,\mathcal H_k)\ \mathbb P(X=x |A_k =a_{\mathrm s},\mathcal H_k)\\
&\stackrel{(c)}{=}  f(Y_k= y|X = x, A_k =a_{\mathrm s}) b_k[x]\\
&\stackrel{(d)}{=}\left[\nu \exp\left\{ -\nu y\right\} \right]^{\delta[a_{\mathrm s},x]}[\exp\{-y\}]^{1-\delta[a_{\mathrm s},x]} b_k[x]
\\
&\stackrel{(e)}{=}\exp\left\{-y+J(y)\delta[a_{\mathrm s},x]\right\}b_k[x],
 \numberthis
\end{align*}
where (a) follows from the definition of belief; (b) follows from Bayes' rule and $\propto$ denotes proportionality up to a normalization factor independent of $x$;  (c) follows from the facts that $Y_k$ is independent of history $\mathcal H_k$ given $(X,A_k)$, 
and $X$ is independent of action $A_k$ given  $\mathcal H_k$,
 and by the definition of belief $b_k$; (d-e) follow by substitution of the pdf of $Y_k$ given in \eqref{eq:feedback_pdf} and by definition of $J(y)$.
We prove the lemma using induction.  The lemma holds for $b_0$ by definition of $m_0$. Let $0 \leq k \leq L-1$ and $b_k$ be given by \eqref{eq:beliefup}, then using 
\eqref{eq:beliefder}(e) normalized to sum to one, we get 
\begin{align*}
b_{k+1}[x] &= \frac{\exp\left\{ J(y){\delta[a_{\mathrm s},x]} \right\} \exp(m_k[x]\}}{\sum_{l \in \mathcal I}\exp\left\{ J(y) {\delta[a_{\mathrm s},l]} \right\}  \exp\{m_k[l]\}} \\
&= \frac{\exp\{m_{k+1}[x]\}}{\sum_{l\in \mathcal I}\exp\{m_{k+1}[l]\} },\numberthis
\end{align*}
where $m_{k+1}[x] $ is given by \eqref{eq:mkup}.
\end{proof}
 Let $\vect m_k \triangleq [m_k[1],\ldots m_k[{|\mathcal I|}]$. Then, the previous lemma demonstrates that $\vect m_k$ is a sufficient statistic for control decisions, since it is sufficient for computing the belief $b_k$ at time $k$. 
 Therefore, $\mu$ can be expressed as $A_k=\mu_k(\vect m_k) ,\ \forall 0 \leq k \leq L$, which maps the current preference vector $\vect m_k$ to beam index $A_{k}\in \mathcal I$.
 This result makes it possible to achieve an efficient implementation, since the belief can be updated according to simple preference update rules as in \eqref{eq:mkup}, rather than via complex Bayesian belief updates.
 In the subsequent analysis, we will use $\vect m_k$ rather than $b_k$ as the state.   
  \vspace{-4mm}
 \subsection{MDP Formulation}
 Thanks to the identification of the sufficient statistic $\vect m_k$,
 we model the optimization problem $\textbf{P1}$ as a Markov decision process (MDP) and optimize the decision variables to maximize the alignment probability in the data-communication phase. The MPD is a 5-tuple $\langle \mathcal T, \mathcal S, \mathcal I, f(\vect m_{k+1}|\vect m_k, a_k ), r_k(\vect m_k, a_k),\forall k{\in}\mathcal T \rangle$, with elements described as follows.\\
\begin{comment}$\mathcal T $ denotes the time horizon of the MDP, $\mathcal S$ denotes the state-space, ${\mathcal I}$ denotes the action-space at time slot $k\in \mathcal T$, $f(\vect m_{k+1}|\vect m_k,  a_k )$ denotes the ensemble of state-transition pdf, $r_k(\vect m_k, a_k): \mathcal S\times {\mathcal I} \to \mathbb R$ denotes the reward in slot $k$ obtained after enacting action $ a_k$ in state $\vect m_k$. We describe the MDP elements for our specific problem as follows.  \end{comment} 
 \textbf{Time Horizon:} given as
 $\mathcal T{=}\{0,1,\ldots,L \}$ where $\mathcal T_{\mathrm{BA}}{\equiv}\mathcal T{\setminus}\{L\}$ denote the slot indices associated with the beam-alignment phase, whereas at $k{=}L$, the communication parameters are selected and used until the end of the frame.
 \\
 \textbf{State space:} given as $
 \mathcal S = \mathbb R^{|\mathcal I|}$, i.e., all possible values of preference vectors $\vect m_k$.
  \\
 \textbf{Action space:} the set containing all the beam indices, {$\mathcal I$}.\\ 
 \textbf{State transition distribution: } 
 Given state $\vect m_k=\vect m$ and action $A_k = a_{\mathrm s}$
 used in the $k$th stage of the beam-alignment phase, 
 the feedback $Y_k=y$ is generated with pdf
\begin{align*}
\label{eq:y_k_pdf}
\numberthis
& f(y|\vect m,a_{\mathrm s})\triangleq
 \sum_{x\in\mathcal I} f(Y_k=y|X= x,A_k=a_{\mathrm s}) b_k[x]\\
 &=
\frac{\exp\{m[a_{\mathrm s}]\}}{\sum_{l\in\mathcal I}\exp\{m[l]\}} \nu e^{-\nu y}+ \left[1{-}\frac{\exp\{m[a_{\mathrm s}]\}}{\sum_{l\in\mathcal I}\exp\{m[l]\}}\right]e^{-y},
\end{align*}
 leading to the new state
\begin{align}
\label{mnew}
 \vect m_{k+1}=\vect m +  J(y) \boldsymbol\delta[a_{\mathrm s}],
\end{align}
where $\boldsymbol\delta[a_{\mathrm s}]{=}[\delta[a_{\mathrm s},x]]_{\forall  x \in\mathcal I}$
 is the vector with entries $\delta[a_{\mathrm s},x]$.
\\
\textbf{Reward function:}
the reward is the probability of choosing a beam index such that {$A_{\mathrm d}=X$}
in the data communication phase, so that correct alignment is achieved, yielding
\begin{align}
\label{reward}
r_k(\vect m,a) =
\begin{cases}
0, &{k\in\mathcal T_{\mathrm{BA}}},\\
\frac{\exp\{m[a]\}}{\sum_{l\in\mathcal I}\exp\{m[l]\}},\  &k=L.
\end{cases}
\end{align}
We now formulate the value function iteration for the MDP.
\vspace{-4mm}
\subsection{Value Function}
The value function under the optimal policy is given as
\begin{align}
\label{eq:vfunq}
V_k^*(\vect m)&= \max_{a_{\mathrm s} \in{\mathcal I}}  q_k(\vect m, a_{\mathrm s}),
\end{align} 
where $q_k$ is the Q-function under the state-action pair $(\vect m,a)$, defined recursively as
\begin{align*}
q_{L}(\vect m,{A_{\mathrm d}}) = r_{L}(\vect m,{A_{\mathrm d}})= \frac{\exp\{m[{A_{\mathrm d}}]\}}{\sum_{l\in\mathcal I}\exp\{m[l]\}}, 
\end{align*}
and for $k\in\mathcal T_{\mathrm{BA}}$, using \eqref{eq:y_k_pdf},
\begin{align*}
\label{eq:qfuncup}
&  q_k(\vect m,{a_{\mathrm s}})
 {=}\int_{\mathbb R^{|\mathcal I|}}\!\!\!\!\!\!V_{k+1}^*(\vect m^\prime) f(\vect m_{k+1} {=}\vect m^\prime|\vect m_k{=}\vect m,A_k{=}a_{\mathrm s}) \mathrm d \vect m^\prime \\
 &= \int_{0}^{\infty} V_{k+1}^*(\vect m + J(y)\boldsymbol\delta[a_{\mathrm s}]) f(y|\vect m,a_{\mathrm s}) \mathrm d y. \numberthis
\end{align*}
This yields the optimal value function in the data communication phase,
by choosing the beam index with maximum preference
${A_{\mathrm d}}^*=\arg\max_{{A_{\mathrm d}}\in\mathcal I} m[{A_{\mathrm d}}]$,
\begin{align}
\label{eq:valfunNB}
V_{L}^*(\vect m) = \max_{A_{\mathrm d} \in {\mathcal I}} q_{L}(\vect m,A_{\mathrm d}) =   \frac{\exp\{m[{A_{\mathrm d}}^*]\}}{\sum_{l\in\mathcal I}\exp\{m[l]\}}.
\end{align}

In the beam-alignment phase ($k\in\mathcal T_{\mathrm{BA}}$), combining \eqref{eq:vfunq} and \eqref{eq:qfuncup}, we obtain iteratively the value function as
\begin{align*}
V_k^*(\vect m)&= \max_{a_{\mathrm s}\in\mathcal I} 
\int_{0}^{\infty} V_{k+1}^*(\vect m + J(y)\boldsymbol\delta[a_{\mathrm s}]) f(y|\vect m,a_{\mathrm s}) \mathrm d y. 
\end{align*} 
In the following theorem, whose proof is provided in the Appendix, we unveil structural properties of $V_k^*(\vect m)$. We find a lower-bound and an upper-bound to the Q-function
and show that these bounds are optimized by a policy which, in each stage of the beam-alignment phase, selects the beam index with the \emph{second-best} preference.
This result will be the basis for our proposed policy evaluated numerically in  Sec. \ref{sec:numer}.

\begin{theorem}
\label{thm1}
For $k\in\mathcal T_{\mathrm{BA}}$, the Q-function is bounded as
\begin{align*}
\label{eq:qfunLB}
&q_k(\vect m,a_{\mathrm s}){\geq}\numberthis
 q_k^{LB}(\vect m,a_{\mathrm s}){\triangleq}
 \frac{1}{\sum_{l\in\mathcal I}\exp\{m[l]\} } 
 \Biggl[ \xi(a_{\mathrm s};\vect m)
\\
&{+}\exp\left\{ \frac{\min_{ x_i\neq x_j} m[ x_i]{-}\nu m[ x_j]}{1-\nu} \right\}h(\nu) 
\frac{g(\nu){-}[g(\nu)]^{L-k}}{1-g(\nu)} \Biggr],\\
&q_k(\vect m,a_{\mathrm s}){\leq }
q_k^{UB}(\vect m,a_{\mathrm s}){\triangleq}
 \frac{\left[ 1+h(\nu) \right]^{L-k-1} }{\sum_{l\in\mathcal I}\exp\{m[l]\} }  \xi(a_{\mathrm s};\vect m),\!\! \numberthis
 \label{eq:qfunUB}
\end{align*}
where we have defined $\xi(a_{\mathrm s};\vect m)$
\begin{align}
\label{eq:xifuc}
\triangleq\!\!\begin{cases}
\exp\{m[a_{\mathrm s}]\},\ \ \ 
\text{if }\max_{{\hat a }\neq a_{\mathrm s}} m[{\hat a }]{-}m[a_{\mathrm s}]{<}\ln \nu,
\\
\exp\{\max_{{\hat a } \neq a_{\mathrm s}} m[{\hat a }]\}\\\ 
+ h(\nu) \exp\left\{\frac{m[a_{\mathrm s}]- \nu \max_{{\hat a }\neq a_{\mathrm s}} m[{\hat a }] }{1-\nu}\right\},
\ 
\text{otherwise,}
\end{cases}\!\!\!\!
\end{align}
where
\begin{align}
&h(\nu) \triangleq \exp\left\{ \frac{\nu}{1-\nu} \ln \nu\right\}- \exp\left\{\frac{\ln \nu}{1-\nu}\right\} >0,\\
&g(\nu) \triangleq \exp\left\{\frac{\ln \nu}{1-\nu}\right\}\left[ \frac{1}{\nu+1}- \frac{\ln\nu}{1 - \nu}\right]>0.
\end{align}

Let $ x_{[1]}, x_{[2]},\dots, x_{[|\mathcal I|]}$ be an ordering of beam indices in decreasing order of preference, i.e.,
$m[ x_{[1]}]\geq m[ x_{[2]}]\geq\cdots,m[ x_{[|\mathcal I|]}]$, then
the optimal value function is bounded as
\begin{align*}
\label{eq:valfun_lb}
&\!\!\!V_k^*(\vect m) \geq 
{\max_{a_{\mathrm s}\in\mathcal I} q_k^{LB}(\vect m,a_{\mathrm s})
=q_k^{LB}(\vect m, x_{[2]})},\ \forall k{\in}\mathcal T_{\mathrm{BA}},\numberthis
%\\&
%\geq\frac{1}{\sum_{l\in\mathcal I}\exp(m[l])}  \Biggl[\exp
%m[x_{[1]}]\}\\
%&{+}\exp\!\left\{\!\frac{\min_{ x_i{\neq}x_j} \{ m[ x_i]{-}\nu m[ %x_j]\}}{1-\nu}\!\right\}\!h(\nu) 
%\frac{1{-}[g(\nu)]^{L-k}}{1-g(\nu)} \Biggr], 
\\
\label{eq:valfun_ub}
&\!\!\!V_k^*(\vect m) \leq  {\max_{a_{\mathrm s}\in\mathcal I} q_k^{UB}(\vect m,a_{\mathrm s})
=q_k^{UB}(\vect m, x_{[2]})},\ \forall k{\in}\mathcal T_{\mathrm{BA}}, \numberthis
%\\
%&=   \frac{e^{\max_{\hat a} m[\hat  n]} }{\sum_{l\in\mathcal I}
%\exp(m[l]) } \Bigl[1+h(\nu)\Bigr]^{L-k}, \numberthis
\end{align*} 
with the maximizer of $q_k^{UB}$ and $q_k^{LB}$ given by the \emph{second-best} beam index $x_{[2]}$.
\end{theorem}
\begin{proof}
 The proof is provided in the Appendix.
\end{proof}
As a result of this Theorem, both the upper and lower bounds of the Q-function
are maximized by the \emph{second-best} beam index policy, which selects the beam index with the second-best preference during the beam-alignment phase. This policy will be evaluated numerically in the next section, against other MAB-based schemes proposed in the literature.
\begin{comment}
\begin{theorem}
For $k \in \mathcal T$, the optimal value function is lower-bounded and upper-bounded as
\begin{align*}
\label{eq:valfun_lb}
&V_k^*(\vect m) \geq \frac{1}{\sum_{l\in\mathcal I}\exp(m[l])}  \Biggl[{\exp\left(\max_{\hat a} m[\hat  n]\right) }\\
& + {\exp\left(\frac{\min_{ x_i\neq x_j} \{ m[ x_i]- \nu m[ x_j]\}}{1-\nu}\right)} h(\nu) \sum_{n=0}^{L-k-1} g^n(\nu) \Biggr], \numberthis
\end{align*}
\begin{align*}
\label{eq:valfun_ub}
&V_k^*(\vect m) \leq  \frac{\exp\left(\max_{\hat a} m[\hat  n]\right) }{\sum_{l\in\mathcal I}\exp(m[l])} \Bigl[1+h(\nu)\Bigr]^{L-k}, \numberthis
\end{align*}
where 
\begin{align}
h(\nu) \triangleq \exp\left( \frac{\nu}{1-\nu} \ln \nu \right)- \exp\left( \frac{\ln \nu}{1-\nu} \right) >0,
\end{align}
\begin{align}
g(\nu) \triangleq \exp\left(\frac{\ln \nu}{1-\nu}\right) \left[ \frac{1}{\nu+1}- \frac{\ln\nu}{1 - \nu}\right]>0.
\end{align}
\end{theorem}
\end{comment}
\vspace{-3mm}
\section{Numerical Results}
\label{sec:numer}
In this section, we evaluate the performance of the \emph{second-best} beam index selection scheme ($a_{\mathrm s}{=}x_{[2]}$) with analog beamforming, and compare it with three other schemes. The first one is based on LTS, a popular MAB scheme \cite{Sutton}. In LTS, at each slot the action is chosen according to the belief distribution, i.e., $a_{\mathrm s}{\sim}b[x]$. The second scheme is based on scanning the most-likely beam index ($a_{\mathrm s}{=}x_{[1]}$) as proposed in \cite{Javidi} (first-best). The third scheme is based on UCB as proposed in \cite{MAB}.
We evaluate the performance of these three schemes in terms of the probability of alignment and spectral efficiency using Monte-Carlo simulation with $10^5$ iterations for each simulated point, with parameters as follows: $M_{\mathrm t}{=}128$, $M_{\mathrm r}{=}1$, $N_0{=}{-}174\mathrm{dBm/Hz}$, $\wtot{=}200\mathrm{MHz}$, $\tfrm{=}20\mathrm{ms}$, $T_s{=}0.1\mathrm{ms}$, $f_c{=}30\mathrm{GHz}$, $d{=}10$m, $[\text{path loss exponent}]{=}2$. 
The BS uses $M_{\mathrm t}{=}128$ antennas and partitions the AoD space into $32$ sectors, each with a beamwidth of 
$\pi/32\mathrm{rad}$ and with uniform prior $b_0[ x] = 1/32,\ \forall  x \in\mathcal I$; the UE is isotropic, hence it uses $M_{\mathrm r}{=}1$ antenna with a single sector. 
We use the beamforming design proposed in \cite{noh} for ULAs with antenna spacing $d_{\mathrm t}{=}\lambda/2$. 
With this configuration, the main-lobe and side-lobe gains are best approximated by $G \approx 14\mathrm{dB}$, $g\approx-11\mathrm{dB}$.
\begin{figure}[!t]
\centering
\includegraphics[width=0.9\linewidth]{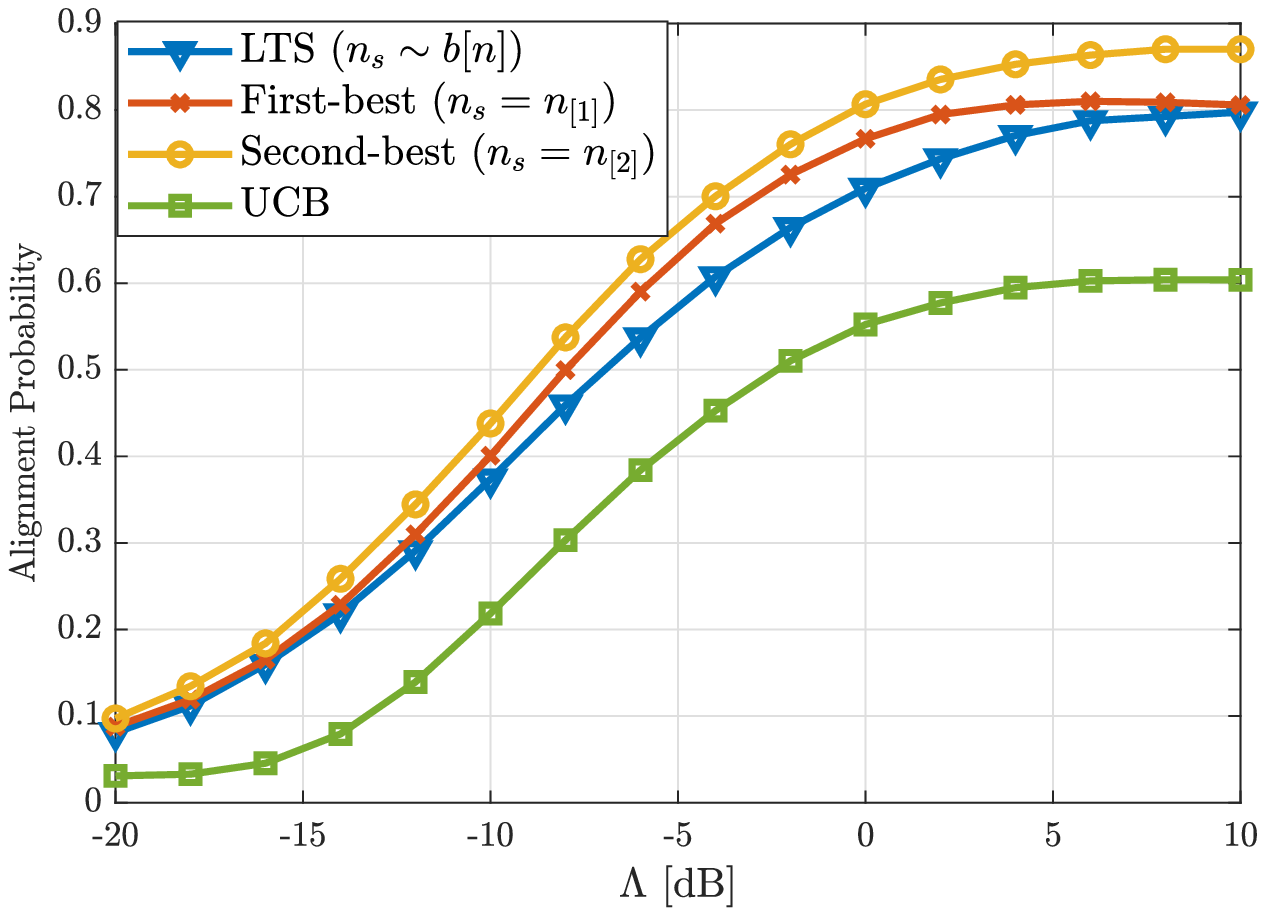}
\caption{Alignment Probability vs $\Lambda$; $L = 32$ (beam-alignment takes $16\%$ of frame duration).}
\label{fig:paln_snr}
\vspace{-8mm}
\end{figure}
\par In Fig.~\ref{fig:paln_snr}, we depict the probability of alignment achieved by the aforementioned schemes versus the pre-beamforming SNR 
$\Lambda$.
 It can be observed that second-best has better performance than the other three schemes, with up to 7\%, 10\%, and 30\% performance gains compared to first-best, LTS-based and UCB-based schemes. The performance gain of second-best is attributed to a better exploration-exploitation trade-off.
  The first-best scheme suffers from poor exploration since it "greedily" chooses the beam index most likely to succeed, but fails to test other beams that may be under-explored, and is thus prone to make alignment errors.
On the other hand, LTS-based scheme suffers from poor exploitation since it may scan least likely beams.       
The proposed \emph{second-best} scheme, on the other hand, strikes a favorable trade-off between exploration and exploitation: instead of greedily choosing the most likely beam, it chooses the second most likely one, leading to better exploration than \emph{first-best}; simultaneously, by not choosing beam pairs that are unlikely to succeed, it leads to  a better exploitation compared to the LTS-based and UCB-based schemes. Finally, compared to UCB,  second-best is better tailored to the structure of the model, since it aims to maximize the alignment probability \emph{at the end} of the beam-alignment phase (see \eqref{reward}), rather than the surrogate metric of UCB -- the cumulative SNR accrued \emph{during} beam-alignment.
%%%%%%%
\par In Fig.~\ref{fig:rate_NB}, we depict the spectral efficiency against the fraction of $\tfrm$ used for BA $L T_{\mathrm s}/\tfrm$. We fix the SNR for beam-alignment as $\Lambda=0\mathrm{dB}$ and the data-communication power as $ P_{\mathrm d}^*{=}22\mathrm{dBm}$. Similar to {Fig.~\ref{fig:paln_snr}}, \emph{second-best} outperforms the three other schemes, owing to improved alignment. The spectral efficiency  is maximized at a unique maximizer $L^*$: it increases initially with $L {\leq} L^* $ as the beam-alignment probability improves with $L$. However,  as $L$ increases beyond $L^*$, this gain is offset by the increased overhead and reduced duration of the data communication phase.  
\begin{figure}[!t]
\centering
\includegraphics[width=0.9\linewidth]{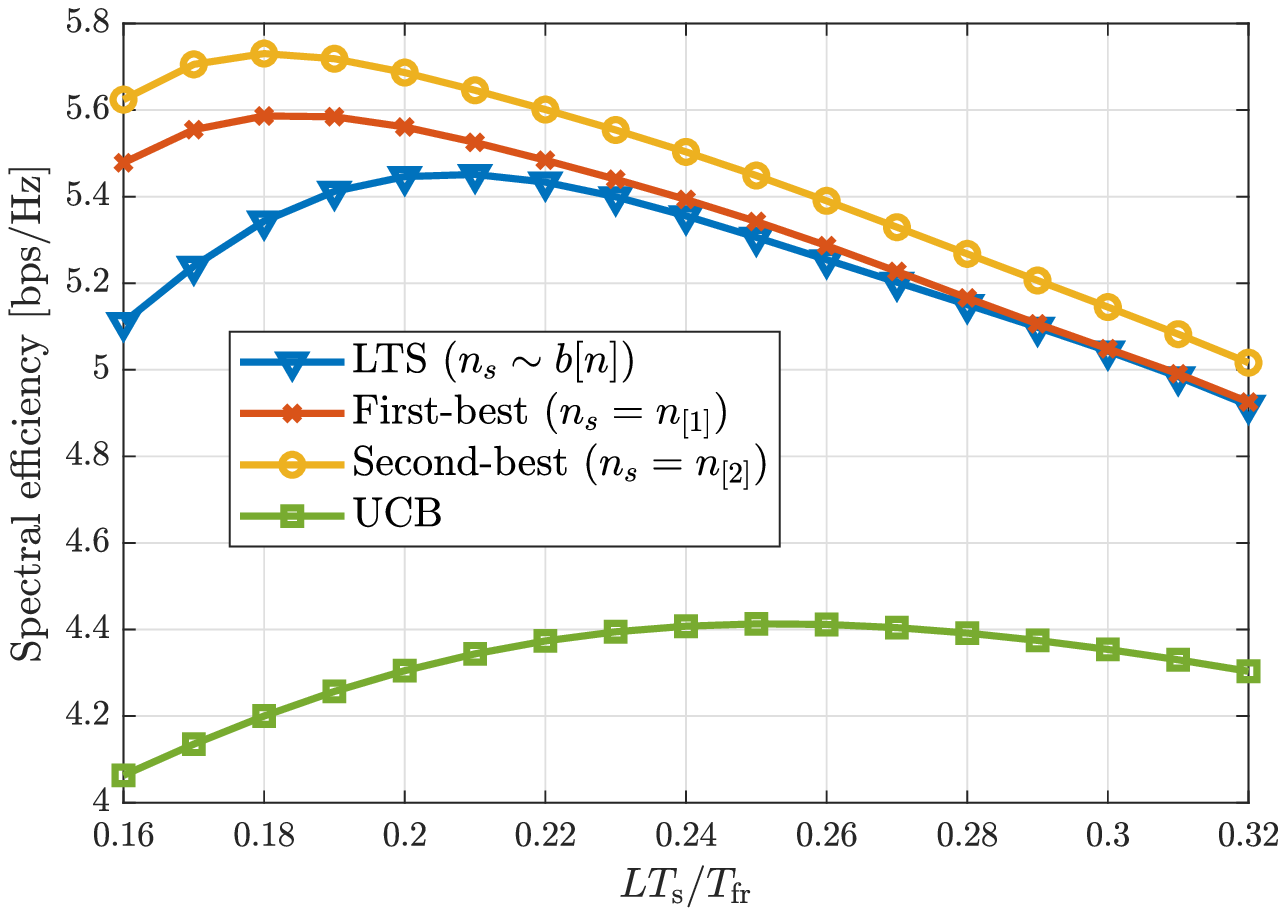}
\caption{Spectral efficiency vs fraction of $\tfrm$ used for BA  $L T_{\mathrm s}/T_{\mathrm{fr}}$.}
\label{fig:rate_NB}
\vspace{-6mm}
\end{figure}
\vspace{-5mm}
{\section{Conclusions}
\label{sec:conc}
In this paper, we have formulated the beam-alignment problem as a Bayesian MAB problem. For the optimal control  
design, we have identified a simplified sufficient statistic referred to as the preference of beam pairs. Based on the preference and bounding of the value function, we have proposed a heuristic policy, which selects the beam pair with the second best-preference to scan. We have shown numerically that the proposed scheme outperforms the first-best, LTS, and UCB based beam-alignment schemes proposed in the literature.   
\vspace{-4mm}
\section*{Appendix: Proof of Theorem \ref{thm1}}
\label{app1}
\begin{proof}
We prove the theorem using induction. Notice that from  the definition of Q-function \eqref{eq:qfuncup} and the optimal value function expression \eqref{eq:valfunNB} for $k=L$ , we get
\begin{align*}
&q_{L-1}(\vect m, a_{\mathrm s}) =  \int_{0}^{\infty}\frac{e^{\max_{\hat a}{ m^\prime[\hat a|\vect m,a_{\mathrm s},y]}} }{\sum_{l\in\mathcal I}e^{
{m^\prime[l|\vect m,a_{\mathrm s},y]}
}} f(y|\vect m,a_{\mathrm s})  dy, 
\end{align*}
where we have defined
the preference update \eqref{mnew} as
$$
m^\prime[ x|\vect m,a_{\mathrm s},y]
=m[x]+J(y)\delta[a_{\mathrm s},x].
$$
Moreover, using \eqref{eq:y_k_pdf} and \eqref{eq:Jfun} we note that
\begin{align}
\label{eq:sumter}
{\sum_{l\in\mathcal I}e^{m^\prime[l|\vect m,a_{\mathrm s},y]}}
&{=}
\sum_{l\in\mathcal I}e^{m[l]{+}J(y)\delta[a_{\mathrm s},l]} 
{=}e^{y}f(y|\vect m{,}a_{\mathrm s}){\sum_{l\in\mathcal I}e^{m[l]} }.
%%%
\end{align}
This yields
\begin{align*}
\label{eq:qfunLS}
q_{L-1}(\vect m, a_{\mathrm s})
&\stackrel{(a)}{=}  \frac{1}{\sum_{l\in\mathcal I}e^{m[l]}} \int_{0}^{\infty}\!\!\!\!\!\!  {e^{\max_{\hat a} m[\hat  a] +J(y) \delta[a_{\mathrm s},\hat a]} } e^{-y} \mathrm d y \\
&\stackrel{(b)}{=}  \frac{1}{\sum_{l\in\mathcal I}e^{m[l]}} \xi(a_{\mathrm s}; \vect m)
\numberthis
\end{align*}
where (b) follows by evaluating the integral in (a) for the two cases in \eqref{eq:xifuc}, and noting that it is given by  $\xi(a_{\mathrm s};\vect m)$.
Using Lemma \ref{lemma2} and \eqref{eq:qfunLS}(b), the optimal value function becomes
\begin{align*}
V_{L-1}^*(\vect m) %&= \frac{1}{\sum_{l\in\mathcal I}e^{m[l]}}  %\max_{a_{\mathrm s}} \xi(a_{\mathrm s}; \vect m)\\
&= \frac{1}{\sum_{l\in\mathcal I}e^{m[l]}} \left[e^{m[ x_{[1]}] } +  h(\nu) e^{\frac{m[ x_{[2]}]- \nu m[ x_{[1]}]}{1-\nu}} \right].
\end{align*}
 Thus, the theorem statement holds for $k{=}L-1$ with equality. Assume it holds for $k{+}1$. 
 {Using Lemma \ref{lemma2},} we can bound
 $$
\max_{\hat a}\xi(\hat a;\vect m^\prime[ n|\vect m,a_{\mathrm s},y] )
 \geq \exp\{ \max_{\hat a} m^\prime[ \hat a|\vect m{,}a_{\mathrm s},y] \}
 $$$$+ h(\nu) e^{ \frac{ \min_{x_i \neq x_j} m^\prime[ x_i|\vect m{,}a_{\mathrm s}{,}y]{-} \nu m^\prime[ x_j|\vect m,a_{\mathrm s},y]}{1-\nu}}.$$
Using {\eqref{eq:qfuncup}}, the induction hypothesis \eqref{eq:valfun_lb} for $k{+}1$ and {the} above bound,  {we obtain}
$q_k(\vect m,{a_{\mathrm s}})$
\begin{align*}
\label{eq:qfun_lb1}
&
%{=\int_{0}^{\infty} V_{k+1}^*(\vect m + J(y)\boldsymbol
%\delta[a_{\mathrm s}]) f(y|\vect m,a_{\mathrm s}) \mathrm d y}\\
%&
%{
%\geq  \frac{1}{\sum_{l\in\mathcal I}e^{m[l]}} 
%\int_{0}^{\infty}e^{-y}\mathrm d y
% \Biggl[
% \max_{\hat a}
%  \xi(\hat a;\vect m + J(y)\boldsymbol\delta[a_{\mathrm s}])
%{+}\exp\left\{ \frac{\min_{ x_i\neq x_j} m^\prime[x_i|\vect %m,a_{\mathrm s},y]{-}\nu m^\prime[ x_j|\vect m,a_{\mathrm s},y]}{1-
%\nu} \right\}h(\nu) 
%\frac{g(\nu){-}[g(\nu)]^{L-k-1}}{1-g(\nu)} \Biggr],
%}
{\geq}\!\!\int_{0}^{\infty}\!\!\Biggl\{ \frac{e^{\max_{\hat a}{ m^\prime[{\hat a}|\vect m,a_{\mathrm s},y]}} }{\sum_{l\in\mathcal I}e^{
{ m^\prime[l|\vect m,a_{\mathrm s},y]}
}} + 
\frac{
e^{
\frac{
\min\limits_{ x_i\neq x_j}
{
 m^\prime[ x_i|\vect m,a_{\mathrm s},y]
-\nu m^\prime[ x_j|\vect m,a_{\mathrm s},y]}}
{1-\nu}}}
{\sum_{l\in\mathcal I}e^{{ m^\prime[l|\vect m,a_{\mathrm s},y]}}} 
\\ 
&\quad\times h(\nu) \frac{1-[g(\nu)]^{L-k-1}}{1-g(\nu)} \Biggr\} f(y|\vect m,a_{\mathrm s})\ \mathrm d y. \numberthis
\end{align*}

Moreover, we note that
\begin{align*}
\label{eq:min_ter}
&\min_{ x_i\neq x_j}
m^\prime[ x_i|\vect m,a_{\mathrm s},y]
-\nu m^\prime[ x_j|\vect m,a_{\mathrm s},y]
\\
%&
%=\min_{ x_i\neq x_j}  m[ x_i]- \nu m[ x_j] + J(y)\{\delta[a_{\mathrm %s}, x_i]-\nu \delta[a_{\mathrm s}, x_j] \} \\
&\geq \min_{ x_i\neq x_j} [m[ x_i]{-} \nu m[ x_j]]{+} \min_{ x_i\neq x_j} J(y)\{\delta[a_{\mathrm s}, x_i]{-}\nu \delta[a_{\mathrm s}, x_j] \}\\
&=\min_{ x_i\neq x_j} [m[ x_i]- \nu m[ x_j]]+\min\{J(y),-\nu J(y)\}.\numberthis
\end{align*}
By substituting \eqref{eq:min_ter} and \eqref{eq:sumter} into \eqref{eq:qfun_lb1}, yields
\begin{align*}
\label{eq:qfun_lb22}
&q_k(\vect m,{a_{\mathrm s}}){\geq}\frac{1}{\sum\limits_{l\in\mathcal I}e^{m[l]} }\! \Biggl[\! \int_{0}^{\infty}\!\!e^{\max_{\hat a} m[\hat a]{+}J(y) \delta[a_{\mathrm s}{,}{\hat a}]}  e^{-y} \mathrm d y \\
&{+}e^{ \frac{\min_{ x_i\neq x_j} m[ x_i]{-}\nu m[ x_j]}{1-\nu}}\!\!\int_{0}^{\infty}e^{\frac{\min\{J(y),-\nu J(y)\}}{1-\nu}}e^{-y} \mathrm d y  
\\ &\quad\times
 h(\nu)\frac{1-[g(\nu)]^{L-k-1}}{1-g(\nu)} \Biggr]. \numberthis
\end{align*}
The first integral in \eqref {eq:qfun_lb22} is equal to $\xi(a_{\mathrm s};\vect m)$ and the second integral is found to be equal to
\begin{align*}
&\int_{0}^{\infty} e^{\frac{\min\{J(y),-\nu J(y)\}}{1-\nu}}e^{-y} \mathrm d y=  e^{\frac{\ln \nu}{1-\nu}} \left[ \frac{1}{\nu{+}1}{-} \frac{\ln\nu}{1 {-} \nu}\right] = g(\nu){>}0.
\end{align*}
Upon substituting these integrals into \eqref{eq:qfun_lb22}  yields the following lower-bound to the Q-function,
\begin{align*}
&
{q_k(\vect m,a_{\mathrm s})}
{\geq}\frac{\xi(a_{\mathrm s};\vect m){+} e^{\frac{\min_{ x_i\neq x_j}\!\!\! m[ x_i]- \nu m[ x_j]}{1-\nu} } h(\nu) \frac{g(\nu){-}[g(\nu)]^{L-k}}{1-g(\nu)} }{\sum\limits_{l\in\mathcal I}e^{m[l]}}\!{,} 
\end{align*}
{which proves the induction step \eqref{eq:qfunLB}, and}
whose maximization (see Lemma \ref{lemma2}) yields \eqref{eq:valfun_lb}.

Similarly, using the induction hypothesis {\eqref{eq:valfun_ub}} for $k+1$ and the upper-bound
\begin{align*}
\max_{\hat a}\xi(\hat a;\vect m^\prime[\vect m,a_{\mathrm s},y] )
 \leq (1+h(\nu)\exp\{ \max_{\hat a} m^\prime[ \hat a|\vect m{,}a_{\mathrm s},y],
\end{align*}
  we obtain the following upper-bound to the Q-function,
\begin{align*}
%\label{eq:qfunub}
&q_k(\vect m,a_{\mathrm s}){\leq}\frac{\left[ 1{+}h(\nu)\right]^{L-k-1} }{\sum_{l\in\mathcal I}e^{m[l]} }\!\!\!\int_{0}^{\infty}\!\!\!{e^{\max_{\hat a} m[\hat  a] +J(y) \delta[a_{\mathrm s}, \hat a]} } e^{-y} \mathrm d y . %\numberthis
\end{align*}
The integral above is equal to $\xi(a_{\mathrm s};\vect m )$,
{which proves the induction step \eqref{eq:qfunUB},} hence
\begin{align*}
\label{eq:vk_ub}
V_k^*(\vect m)
{=}\max_{a_{\mathrm s}\in\mathcal I}q_k(\vect m,a_{\mathrm s})
 {\leq}\frac{[1{+}h(\nu)]^{L-k-1}\!\!\!\!\!\!\!\!\!\!\!\!\!\!}{\sum_{l\in\mathcal I}e^{m[l]} } \max_{a_{\mathrm s} \in\mathcal I}\xi(a_{\mathrm s};\vect m). \numberthis
\end{align*}
Noting that $\max_{a_{\mathrm s}} \xi(a_{\mathrm s};\vect m) = \xi(x_{[2]};\vect m)$ (see Lemma \ref{lemma2}), and  upon substitution in \eqref{eq:vk_ub} yields \eqref{eq:valfun_ub}. 
\end{proof}
\vspace{-5mm}
\begin{lemma}
\label{lemma2}
We have that $\arg\max_{a_{\mathrm s} \in \mathcal I} \xi(a_{\mathrm s};\vect m) =  x_{[2]}$ and
%Let $\xi(a_{\mathrm s};\vect m)$ be the function defined in \eqref{eq:xifuc}, then
\begin{align}
\label{eq:xifunmax}
\!\!\!\max_{a_{\mathrm s}\in \mathcal I} \xi(a_{\mathrm s};\vect m){=}e^{m[ x_{[1]}] }{+}h(\nu) e^{\frac{m[ x_{[2]}]{-}\nu m[ x_{[1]}]}{1-\nu}}.\!
\end{align}
\end{lemma}
\begin{proof}
To show that $\arg\max_{a_{\mathrm s} \in\mathcal I} \xi(a_{\mathrm s};\vect m) =  x_{[2]}$, we proceed as follows. Clearly, if $a_{\mathrm s}\in \{ x_{[2],}x_{[3]},\ldots, x_{[|\mathcal I|]} \}$, then 
$\max_{{\hat a }\neq a_{\mathrm s}} m[{\hat a }]{-}m[a_{\mathrm s}]
=m[x_{[1]}]{-}m[a_{\mathrm s}]\geq 0>\ln(\nu)$, hence
\begin{align*}
\xi(a_{\mathrm s};\vect m){=}e^{m[x_{[1]}]}{+}h(\nu) e^{\frac{m[a_{\mathrm s}]- \nu m[x_{[1]}] }{1-\nu}},
% \geq \xi( x_{[j]};\vect m), j=3,\ldots, |\mathcal I|.
 \end{align*} 
maximized at $a_{\mathrm s}{=}x_{[2]}$.
 Therefore, we restrict $a_{\mathrm s} \in \{  x_{[1]}, x_{[2]}\}$ without loss in performance. Next, we show that $\xi( x_{[2]};\vect m)\geq\xi( x_{[1]};\vect m)$.
  {Let $\Delta{\triangleq} m[ x_{[1]}]{-}m[ x_{[2]}]$. If $\Delta{>}{-}\ln \nu$, then} $\xi( x_{[1]};\vect m)=e^{m[x_{[1]}]}$ and $\xi( x_{[2]};\vect m) > \xi( x_{[1]};\vect m)$.
Otherwise, % ($0 < \Delta\leq -\ln \nu$), we find that
 \begin{align*}
 \xi( x_{[2]};\vect m){-}\xi( x_{[1]};\vect m){\propto} 
\frac{e^{\Delta}{-}1}{ e^{\frac{\Delta}{1{-}\nu}}{-} e^{-\nu\frac{\Delta}{1{-}\nu}}}
{-} h(\nu)\triangleq u(\Delta,\nu).
 \end{align*}
 {Note that $u(\Delta,\nu)$ is decreasing in $\Delta{\in}(0,-\ln \nu],{\forall} \nu {\in} (0,1)$, minimized at 
 $\Delta{=}-\ln \nu$, yielding, after algebraic steps,}
 {
  \begin{align*}
 \xi( x_{[2]};\vect m){-}\xi( x_{[1]};\vect m){\propto} u(\Delta,\nu)\geq
\frac{h(\nu)
}{ e^{-\frac{1{+}\nu}{1{-}\nu}\ln \nu}{-} 1}>0.
 \end{align*}
 }
In both cases, $\max_{a_{\mathrm s}} \xi( a_{\mathrm s} ) = \xi( x_{[2]})$. Upon substitution of $a_{\mathrm s} =  x_{[2]}$ in \eqref{eq:xifuc}, yields \eqref{eq:xifunmax}.
\end{proof}
%%
%\vspace{-2mm}
\bibliographystyle{IEEEtran}
\bibliography{IEEEabrv,biblio} 

\end{document}